\documentclass[3p,preprint]{elsarticle}

\usepackage{amssymb}

\usepackage[figuresright]{rotating}
\usepackage{subcaption}

\newtheorem{theorem}{Theorem}

\newtheorem{proposition}[theorem]{Proposition}
\newtheorem{corollary}[theorem]{Corollary}
\newtheorem{definition}[theorem]{Definition}

\newtheorem{example}{Example}
\newtheorem{remark}{Remark}

\newproof{proof}{Proof}
\newproof{proofsketch}{Proof Sketch}

\usepackage[utf8]{inputenc}
\usepackage{complexity}
\usepackage[algosection]{algorithm2e}
\usepackage{endnotes}
\usepackage{listings}
\usepackage{mathtools}
\usepackage{verbatim}
\usepackage{mathrsfs}
\usepackage{longtable}
\usepackage{enumitem}
\usepackage{tikz}

\usepackage{etoolbox}
\usepackage{float}
\usepackage{color}
\usepackage{hyperref}

\usepackage{complexity}
\usepackage{comment}
\usepackage{longtable}

\usepackage{iftex}

\ifpdf
  \usepackage{underscore}         
  \usepackage[T1]{fontenc}        
\else
  \usepackage{breakurl}           
\fi

\usetikzlibrary{matrix}

\newcommand{\emr}{{\,\triangleleft}}
\newcommand{\eml}{{\triangleright\,}}

\usepackage{mathtools}

\DeclareMathOperator{\add}{{\rm add}}

\newcommand{\la}{{\leftarrow}}
\newcommand{\ra}{{\rightarrow}}
\newcommand{\head}{{\text{\rotatebox[origin=c]{180}{$\Rsh$}}}}
\newcommand{\blank}{{\sqcup}}

\newclass{\NCM}{NCM}
\newclass{\NCCM}{NCCM}
\newclass{\NPCM}{NPCM}
\newclass{\NCACM}{NCACM}
\newclass{\DPCM}{DPCM}
\newclass{\NPDA}{NPDA}
\newclass{\DPDA}{DPDA}
\newclass{\NFA}{NFA}
\newclass{\DFA}{DFA}
\newclass{\NQCM}{NQCM}
\newclass{\DCM}{DCM}
\newclass{\NTM}{NTM}
\newclass{\DTM}{DTM}
\newclass{\NRBQA}{NRBQA}
\newclass{\NRBSA}{NRBSA}
\newclass{\NRBTA}{NRBTA}
\newclass{\NRBQCM}{NRBQCM}
\newclass{\NRBSCM}{NRBSCM}
\newclass{\NRBTCM}{NRBTCM}
\newclass{\NSA}{NSA}
\newclass{\pd}{pd}

\newclass{\onetDPDA}{1tDPDA}

\newclass{\ftNTM}{ftNTM}
\newclass{\fvNTM}{fvNTM}
\newclass{\fcNTM}{fcNTM}
\newclass{\fvDTM}{fvDTM}
\newclass{\fcDTM}{fcDTM}

\newclass{\NTCM}{NTCM}
\newclass{\DTCM}{DTCM}

\newclass{\ftNTCM}{ftNTCM}
\newclass{\fvNTCM}{fvNTCM}
\newclass{\fcNTCM}{fcNTCM}
\newclass{\fvDTCM}{fvDTCM}
\newclass{\fcDTCM}{fcDTCM}

\newclass{\DCOUNTER}{DCOUNTER}
\newclass{\COUNTER}{COUNTER}

\newcommand{\LL}{\mathcal{L}}
\newcommand{\MM}{\mathcal{M}}
\renewcommand{\SS}{\mathcal{S}}

\newcommand\pre{{\rm pre}}
\newcommand\post{{\rm post}}
\newcommand\conf{{\rm conf}}

\begin{document}

\begin{frontmatter}




\title{Store Languages of Turing Machines and Counter Machines}

\author[label1]{Noah Friesen}
\address[label1]{Department of Mathematics and Statistics\\
University of Saskatchewan\\
Saskatoon, SK S7N 5C9, Canada}
\ead{noah.friesen@usask.ca}

\author[label2]{Oscar H. Ibarra}
\address[label2]{Department of Computer Science\\ University of California, Santa Barbara, CA 93106, USA}
\ead{ibarra@cs.ucsb.edu}

\author[label3]{Jozef Jir\'asek\fnref{fn1}}
\address[label3]{Department of Computer Science, University of Saskatchewan\\
Saskatoon, SK S7N 5C9, Canada}
\ead{jirasek.jozef@usask.ca}
\fntext[fn1]{Supported, in part, by Natural Sciences and Engineering Research Council of Canada Grant 2022-05092 (Ian McQuillan)}

\author[label3]{Ian McQuillan\corref{corr}\fnref{fn1}
\corref{corr}}
\ead{mcquillan@cs.usask.ca}
\cortext[corr]{Corresponding author}





\begin{abstract}
The \emph{store language} of an automaton is the set of store configurations (state and store contents, but not the input) that can appear as an intermediate step in an accepting computation.
A one-way nondeterministic finite-visit Turing machine ($\fvNTM$) is a Turing machine with a one-way read-only input tape, and a single worktape, where there is some number $k$ such that in every accepting computation, each worktape cell is visited at most $k$ times.
We show that the store language of every $\fvNTM$ is a regular language. 
Furthermore, we show that the store language of every $\fvNTM$ augmented by reversal-bounded counters can be accepted by a machine with only reversal-bounded counters and no worktape. Several applications are given to problems in the areas of verification and fault tolerance, and to the study of right quotients. We also continue the investigation of the store languages of one-way and
two-way machine models where we present some conditions under which their
store languages are recursive or non-recursive.
\end{abstract}

\begin{keyword}
Store Languages \sep Finite-Visit Turing Machines \sep Decidability Questions \sep Counter Machines
\end{keyword}

\end{frontmatter}

\section{Introduction}
\label{sec:intro}

A useful concept in formal language and automata theory is that of the store language of a machine.
The \textit{store language} is the set of store configurations (state and store contents) that can appear as an intermediate step in an accepting computation.
For example, the store language of a one-way nondeterministic pushdown automaton consists of  all words of the form $q\gamma$, where there is some accepting computation that passes through a configuration in state $q$ and with $\gamma$ as the pushdown contents.
It is well-known that the store language of any pushdown automaton is a regular language \cite{GreibachCFStore,CFHandbook}.
Furthermore, it is known that a nondeterministic finite automaton accepting the store language of a given pushdown automaton can always be constructed in polynomial time \cite{NPDAStoreAlgorithm}.
In addition, the optimal number of states needed in the nondeterministic finite automaton constructed from the given pushdown automaton has been determined \cite{NPDAStoreDescriptional}.

There are several other models where the store languages have been characterized as well.
One-way stack automata are one-way pushdown automata with the additional ability to read from the inside of the stack in a two-way read-only mode, but pushing and popping are only available when at the top of the stack \cite{StackAutomata}; 
finite-turn Turing machines (denoted here by $\ftNTM$) have a one-way read-only 
input tape and a single worktape where there is a bound on the number of changes in direction of the worktape head on the worktape \cite{visitautomata};
and finite-flip pushdown automata are pushdown automata which can ``reverse'' (or flip) their pushdown store up to a bounded number of times \cite{flipPushdown}. 
To note, finite-turn Turing machines are also known as reversal-bounded Turing machines \cite{visitautomata} but we use the term finite-turn as we use the term reversal-bounded for counters instead.
It has been proven that the store languages of all of these models are also always regular (\cite{KutribCIAA2016,StoreLanguages} for stack automata, \cite{StoreLanguages} for $\ftNTM$, 
and \cite{IbarraMcQuillanVerification} for finite-flip pushdown automata). 
One-way counter machines have some finite number of counters that each contains a non-negative integer, and transitions can test whether the value in each counter is zero or not. While it is well known that these machines can accept all recursively enumerable languages \cite{HU}, restrictions can limit their power.
A counter machine is \textit{reversal-bounded} if in every accepting computation, for each counter, the number of switches between executing sequences of transitions that do not decrease the counter and transitions that do not increase the counter, is bounded \cite{Ibarra1978,Baker1974}.
The set of one-way nondeterministic reversal-bounded multi-counter machines is denoted by $\NCM$, and the set of deterministic machines in $\NCM$ is denoted by $\DCM$.
 The store languages of $\NCM$s can be non-regular, but can always be accepted by $\DCM$s \cite{StoreLanguages}. If we augment any of pushdown automata \cite{StoreLanguages}, finite-turn Turing machines \cite{StoreLanguages}, or finite-flip pushdown automata \cite{IbarraMcQuillanVerification} with any number of reversal-bounded counters, then their store languages can always be accepted by $\NCM$s.
Lest one speculate that this would be true for any model that only produces regular store languages, it is not so as stack automata augmented by one reversal-bounded counter can produce store languages that cannot be accepted by a  $\NCM$ \cite{StoreLanguages} despite stack automata only having regular store languages.

When all store languages of a one-way nondeterministic machine model $\MM$ (under some simple conditions described in the paper) are regular, this implies that the languages accepted by the deterministic machines in $\MM$ are closed under right quotient with regular languages \cite{StoreLanguages}. This technique can completely replace the often lengthy and ad-hoc proofs of closure under right quotient with regular languages for deterministic classes of machines in the literature (e.g.\ for deterministic pushdown automata \cite{GinsburgDPDAs} or deterministic stack automata \cite{DetStackQuotient}).

Also, a class of machines either having store languages that are regular or that can be accepted by $\NCM$ machines has applications to problems in the area of model checking, reachability, and verification. Two commonly studied operations are: given a set of store configurations $C$, the set of store configurations that can follow in zero or more moves of a machine $M$ from those in $C$ is called $\post_M^*(C)$, and the set of store configurations that can lead in zero or more moves to those in $C$ is called $\pre_M^*(C)$.
It is known that given a pushdown automaton $M$ and a regular set of configurations $C$, $\pre_M^*(C)$ and $\post_M^*(C)$ are both regular languages \cite{PushdownVerification}. It is also known that for an $\NCM$ $M$ and a set $C$ of configurations accepted by an $\NCM$, both $\pre_M^*(C)$ and $\post_M^*(C)$ can be accepted by $\DCM$s \cite{IbarraSu}. These operations have also been studied for other machine models, e.g.\ \cite{multipushdownmodel,Seth,MultiStackModelChecking,Bouajjani,Finkel2000}.
More recently, in  \cite{IbarraMcQuillanVerification}, connections were made between store languages and these operations, whereby determining that the store languages of a class of machines are always in certain families of languages implies that $\pre_M^*(C)$ and $\post_M^*(C)$ are in the same family, and vice versa.
More specifically, it was shown that under some simple conditions, any class of machines $\MM$ with only regular store languages always satisfies the following: given a regular set $C$, then both $\pre_M^*(C)$ and $\post_M^*(C)$ are also regular; further, for a class where the store languages can always be accepted by an $\NCM$, then given $C$ accepted by an $\NCM$, both $\pre_M^*(C)$ and $\post_M^*(C)$ can be accepted by an $\NCM$. Also, consider the following problem, called the {\em common configurations problem for a class of machines $\MM$}: ``given two machines $M_1, M_2 \in \MM$, are there any non-initial store configurations in common between $M_1$ and $M_2$?''.  The same paper \cite{IbarraMcQuillanVerification} shows that as long as store languages of a class can always be accepted by an $\NCM$, then the common configurations problem is always decidable. This problem is related to fault-tolerance: if one machine, $M_1$, has a store language that describes all faulty configurations, then computations of the other machine $M_2$ may lead to a fault configuration if and only if they have configurations in common.

In this paper, we are interested in one-way nondeterministic Turing machines with a single two-way read/write worktape. 
Such a machine is \emph{$k$-turn} (respectively $k$-visit, $k$-crossing) if, in every accepting computation, the number of changes in direction of movement of the read/write from left-to-right or right-to-left (respectively the number of visits to a worktape cell, or the number of crosses between the boundary of any two adjacent worktape cells) is at most $k$. It is \textit{finite-turn} (respectively finite-visit, finite-crossing), if it is $k$-turn (respectively $k$-visit, $k$-crossing) for some $k$. The class of one-way finite-turn nondeterministic (respectively finite-visit, finite-crossing) Turing machines is denoted by $\ftNTM$ (respectively $\fvNTM, \fcNTM$). All $\ftNTM$ machines are $\fvNTM$ (but not vice versa), and all $\fvNTM$ machines are $\fcNTM$ but not vice versa (since staying on a cell counts as a visit but not a cross). However, it is known that the languages accepted by $\ftNTM$ are properly contained in the languages accepted by $\fvNTM$, which are equal to the languages accepted by $\fcNTM$. Indeed, the languages accepted by $\fcNTM$ and $\fvNTM$ are closed under
Kleene-* but those accepted by $\ftNTM$ are not \cite{visitautomata}; 
also $\fvNTM$ and $\fcNTM$ precisely characterize an important family that has been studied in the formal language theory literature. Greibach showed \cite{visitautomata} that the languages accepted by $\fvNTM$ and $\fcNTM$ are equal to the languages generated by {\em absolutely parallel grammars} \cite{AbsolutelyParallelGrammars}. And it was shown by Latteux \cite{Latteux} that the languages generated by absolutely parallel grammars are identical to the languages generated by many different types of grammar models that are so-called {\em finite-index}
\cite{RozenbergFiniteIndexGrammars}, where there is some integer $k$ such that every word generated by the grammar has a derivation with at most $k$ nonterminals. The languages generated by finite-index grammars of the following types (which we will not define here) were all shown to coincide: ET0L systems, EDT0L systems (two types of Lindenmayer systems), context-free programmed grammars, ordered grammars, matrix grammars, and matrix grammars with appearance checking \cite{RozenbergFiniteIndexGrammars}. Thus,
finite-crossing and finite-visit Turing machines provide an automata model that accepts exactly the same languages as all of these grammar systems can generate. Hence, $\fvNTM$ and $\fcNTM$ are an important and natural class of machines. 

While it is already known that the store languages of all $\ftNTM$ are regular  \cite{StoreLanguages},
in this paper, we show that the store languages of the more powerful finite-visit and finite-crossing Turing machines are regular languages. This immediately shows that the languages accepted by deterministic finite-visit and deterministic finite-crossing Turing machines are closed under right quotient with regular languages, which was not previously known.  Furthermore, if we augment finite-crossing Turing machines with reversal-bounded counters, then the store languages of all such machines can be accepted by $\NCM$s. These store language results have applications to the $\pre^*$ and $\post^*$ operations, and the common configuration problem, with the common configurations problem being decidable for $\fcNTM$ augmented with reversal-bounded counters. 

Next, we investigate decidability problems, such as, given Turing machine $M$ and $k \ge 0$, is $M$ $k$-turn, (respectively $k$-visit, $k$-crossing)? Also, given $M$, is $M$ finite-turn (respectively finite-visit, finite-crossing)? This would thereby guarantee that its store language is a regular language. Then, we describe the subtle difference between defining finite-crossing as ``in every accepting computation, each boundary between a pair of adjacent cells is crossed at most $k$ times'', and what we call {\em weak finite-crossing}: ``for every accepted word $w$, there is some accepting computation in which the boundary  of each boundary between a pair of adjacent cells is crossed at most $k$ times''. We use the former notion here, but it is the latter notion that is called finite-crossing in \cite{visitautomata}. Despite the two notions being equivalent in terms of languages accepted, this is not so for store languages. The store languages of the first class of machines are all regular languages, however we show that for the second notion, 
the store languages contain non-recursive languages. This shows that there are non-recursive store languages for these machines. This shows that even though two classes of machines accept the same family of languages, they can be dramatically different in terms of store languages. The subtleties of the machine model definition are crucial. 

Finally, we continue the study of store languages of classes of two-way machines. Previously, very little was known besides some examples of non-regular store languages being demonstrated \cite{StoreLanguages}.
We show that several models have non-recursive store languages, including two-way deterministic pushdown automata, and two-way deterministic $1$-counter machines. We also find other classes that have store languages with an undecidable membership problem.

\section{Preliminaries}

We denote the set of integers by $\mathbb{Z}$, the set of non-negative integers by $\mathbb{N}_0$, and the set of positive integers by $\mathbb{N}$. Given $n \in \mathbb{N}_0$, let $\pi(n)$ be $0$ if $n = 0$ and $1$ otherwise.

We assume a familiarity with the basics of formal language and automata theory, including one-way and two-way nondeterministic finite automata, regular languages, and Turing machines \cite{HU}.
An \emph{alphabet} $\Sigma$ is a finite set of symbols.
The \emph{empty word} is denoted by $\lambda$.
Given a word $w\in\Sigma^*$, the \emph{length} of $w$ is denoted by $\left|w\right|$.
Given a word $w = a_1 \cdots a_n, a_i \in \Sigma, 1 \le i \le n$, the {\em reverse} of $w$, $w^R = a_n \cdots a_1$.
The set of all non-empty words over $\Sigma$ is denoted by $\Sigma^+$, and the set of all words over $\Sigma$, including the empty word, is denoted by $\Sigma^*$.
A \emph{language} over $\Sigma$ is any subset of $\Sigma^*$.
Given $L_1, L_2 \subseteq \Sigma^*$, the left quotient of $L_2$ by $L_1$, 
$L_1^{-1}L_2 = \{ v \mid uv \in L_2, u \in L_1\}$, and the right quotient of $L_1$ by $L_2$,
$L_1L_2^{-1} = \{ u \mid uv \in L_1, v \in L_2\}$.
A language $L \subseteq \Sigma^*$ is {\em bounded} if $L\subseteq w_1^* \cdots w_n^*$ for some non-empty
words $w_1, \ldots, w_n \in \Sigma^*$; and $L$ is letter-bounded if $L \subseteq a_1^* \cdots a_n^*$ where $a_1, \ldots, a_n$ are letters of $\Sigma$.

In this paper, we are going to primarily focus on three types of automata: multi-counter machines, Turing machines, and Turing machines augmented by counters. All of the machines are assumed to be nondeterministic, and use a one-way read-only input tape. Hence, we will simply define the model combining together both storage types, and plain Turing machines will be restricted to only use the worktape, and counter machines will only use the counters.

\begin{definition}
A one-way nondeterministic Turing machine with $t$ counters is denoted by a tuple
$M = (Q, \Sigma, \Gamma, \delta, q_0, F)$, where
$Q$ is the finite set of states, $\Sigma$ is the finite input alphabet,
$\Gamma$ is the finite worktape alphabet containing the blank symbol $\blank$, $q_0 \in Q$ is the initial
state, $F \subseteq Q$ is the set of final states,
and
$\delta$ is a finite subset of
$\Omega_0 \cup \cdots \cup \Omega_{t}$, where
\begin{eqnarray*}
\Omega_0 = Q \times (\Sigma \cup \{\lambda,\emr\}) \times \{0\} \times \Gamma \times Q \times \Gamma \times \{{\rm L}, {\rm S}, {\rm R}\},\\
\Omega_i =  Q \times (\Sigma \cup \{\lambda,\emr\}) \times \{i\} \times \{0,1\} \times Q \times \{ -1,0, +1 \} \mbox{~for~} 1 \le i \le t.\\
\end{eqnarray*}


A {\em configuration} of $M$ is a tuple
\begin{equation} (q, w, x \head y, z_1, \ldots, z_t),
\label{config}
\end{equation}
where $q \in Q$ is the current state,
$w \in \Sigma^*\emr \cup \{\lambda\}$ is the remaining input,  
$x \head y$ is the worktape contents with $x \in \Gamma^+, y \in \Gamma^*$, ($\head$ is a new symbol denoting the position of the read/write head, which is scanning the symbol immediately to its left), 
  and $z_i \in \mathbb{N}_0$ is the current contents of counter $i$ for $1 \le i \le t$. 
The {\em store configuration} of the configuration (\ref{config}) is the string $q x \head y C_1^{z_1} \cdots C_t^{z_t}$ where $C_1, \ldots, C_t$ are new symbols, and $\conf(M)$ is the set of all store configurations.
Configurations will change as follows.
\begin{itemize}
\item $(q,aw,x c \head y, z_1,\ldots, z_t) \vdash_M (q', w, x d \head y, z_1,\ldots, z_t)$, if
$(q, a, 0,c,  q',d, {\rm S}) \in \delta,$ (a stay transition on the worktape),
\item $(q,aw,x c \head y, z_1,\ldots, z_t) \vdash_M (q', w, x'  \head d' y, z_1,\ldots, z_t)$, if
$(q, a, 0,c, q',d, {\rm L}) \in \delta , (x = \lambda \implies x' = \blank, \mbox{~otherwise~} x' = x), (y = \lambda \wedge d = \blank \implies d' =\lambda, \mbox{~otherwise~} d' = d)$, (a left transition on the worktape),
\item $(q,aw,x c \head y,z_1,\ldots, z_t) \vdash_M (q', w, x d' c' \head y', z_1,\ldots, z_t)$, if
$(q, a, 0,c ,q', d, {\rm R}) \in \delta, (y = \lambda \implies y' = \lambda, c' = \blank, \mbox{~otherwise~} y = c'y', c' \in \Gamma ), (x = \lambda \wedge d = \blank \implies d' =\lambda, \mbox{~otherwise~} d' = d)$, (a right transition on the worktape),
\item $(q,aw,x c \head y,z_1,\ldots, z_t) \vdash_M (q', w, x c \head y, z_1,\ldots, z_{i-1}, z_i+z, z_{i+1}, \ldots, z_t)$, if
$(q, a, i, \pi(z_i), q',z) \in \delta$ where $1 \le i \le t$, (a counter transition).
\end{itemize}
A {\em computation} of $M$ on $w$ is a derivation
$(p_0,w_0,\gamma_0,z_{0,1}, \ldots, z_{0,t}) \vdash \cdots \vdash (p_n,w_n,\gamma_n, z_{n,1}, \ldots, z_{n,t})$
where $p_0 = q_0, w_0 = w \emr$, and $\gamma_0 = \blank \head$.
This computation is {\em accepting} if $p_n \in F$ and $w_n =\lambda$.
Given such a computation, the {\em worktape address} of configuration $j$, $c_j = (p_j,w_j,\gamma_j, z_{j,1}, \ldots, z_{j,t})$,
for $0 \le j \le n$, denoted $\add(c_j)$ is defined inductively to be $1$ if $j = 0$; and otherwise it is the address of $c_{j-1}$
plus $1$ (resp.\ $-1$, $0$) if $c_{j-1} \vdash c_j$ with a transition that moves right on the worktape (resp.\ moves left on the worktape, stays or uses a counter).
Let $\vdash_M^*$ be the reflexive and transitive closure of $\vdash_M$.

The language accepted by $M$, denoted by $L(M)$, is defined to be the set of all $w \in \Sigma^*$
such that there is an accepting computation of $M$ on $w$.
Furthermore, the store language of $M$ is defined to be:
$$S(M) = \left\{  qx C_1^{z_1} \cdots C_t^{z_t} \left|  \begin{aligned}\  & (q_0, w\emr,  \blank \head, 0, \ldots, 0 ) \vdash_M^*
(q, w', x, z_1, \ldots, z_t) \vdash_M^* (q_f, \lambda, x', z_1', \ldots, z_t'), \\
& q_f \in F,w \in \Sigma^*, w' \in \Sigma^*\emr \cup \{\lambda\}.\end{aligned} \right.\right\}$$
\end{definition}

Intuitively, $\Omega_0$ contains all possible transitions involving the worktape, and $\Omega_i$ contains all possible transitions involving the $i$th counter, for $1 \le i \le t$. 
Therefore, each transition only uses one store at a time (either the worktape, or one of the counters), and it would require multiple transitions to read from multiple stores.
Here the symbol $\emr$ is the
right end-of-input symbol. We will sometimes omit the end-marker for one-way nondeterministic machines, as these machines can guess that they are at the end of the input, as is common for say one-way nondeterministic pushdown automata.
The first component is always the current state, the second component is the input letter being read (or the empty word, or the end-marker), and the third component is the current store being accessed. For the worktape store (so a transition in $\Omega_0$), the fourth component is the current symbol under the read/write head of the store, the fifth component is the state to switch to, the sixth component is the worktape symbol to replace the current symbol, and the seventh component is either ${\rm L}, {\rm S}, {\rm R}$ which provides the direction to move the read/write head on the worktape (left, stay, or right respectively).
For the counter stores in $\Omega_i, i \ge 1$, the fourth component is $0$ if counter $i$ is currently zero and $1$ otherwise, the fifth component is the state to switch to, and the sixth component is the value to add to counter $i$.

Notice that in such a machine with $0$ counters, the store language only encodes the state and worktape and does not use the symbols $C_1, \ldots, C_t$. In such a case, these are simply one-way nondeterministic Turing machines, which we refer to by $\NTM$. In this case, we leave off the third component of the transitions as they are all $0$.
Notice that the read/write head
is encoded within the worktape contents encoded in the store language.
We sometimes examine the restriction where we leave off the worktape and only have counters. We denote by $\COUNTER(t)$ machines with $t$ counters, and $\COUNTER$ all machines with some number of counters.


Next, we study a restriction of the counters, and then three different restrictions of the worktape.
It is well known that Turing machines with one worktape accept all recursively enumerable languages, and likewise one-way deterministic machines with $2$ counters also accept all recursively enumerable languages \cite{HU}. Hence, restrictions are necessary for certain problems to be decidable or to have simpler store languages. All restrictions studied will limit their power.

A one-way nondeterministic Turing machine with $t$ counters $M = (Q, \Sigma, \Gamma, \delta, q_0, F)$ is 
 {\em $r$-reversal-bounded} if, in each accepting computation, the number of changes between sequences of non-decreasing transitions (where transitions can add $1$ or $0$) and sequences of non-increasing transitions (where transitions can add $-1$ or $0$), or vice versa, in each counter is at most $r$. The machine is {\em reversal-bounded} if it is $r$-reversal-bounded for some $r$. The class of all one-way nondeterministic Turing machines with some number of reversal-bounded counters is denoted by $\NTCM$. If the machine does not use the worktape, so $\Omega_0 \cap \delta = \emptyset$, then the class of all one-way nondeterministic machines with some number of reversal-bounded counters (no worktape) is denoted by $\NCM$.

All three restrictions of the worktape were studied in \cite{visitautomata} (finite-turn was called reversal-bounded in their paper, but we use finite-turn here to disambiguate with our use of counters). 
\begin{definition}
Let $M = (Q, \Sigma, \Gamma, \delta, q_0, F)$ be an $\NTM$ (respectively an $\NTCM$).
Let
$(p_0,w_0,\gamma_0, 0, \ldots, 0) \vdash \cdots \vdash (p_n,w_n,\gamma_n, z_1, \ldots, z_n),$ be an
accepting computation, where $i_j$ is the address of the $j$th configuration, for $ 0 \le j \le n$. We say the accepting computation
{\em makes a turn} in the $j$th configuration where $j \ge 1$, if either $i_j < i_{j-1}$ and 
there exists a largest $1 \le l \le j-1$ such that $i_l \ne i_{j-1}$ which has $i_l < i_{j-1}$; or
$i_j > i_{j-1}$ and there exists a 
largest $1 \le l \le j-1$ such that $i_l \ne i_{j-1}$ which has $i_l > i_{j-1}$.
Then, we say the accepting computation is
\begin{itemize}
\item \emph{$k$-turn} if the number of turns is at most $k$,
\item \emph{$k$-crossing} if,
for each $r \in \mathbb{Z}$, $| \{j \mid  i_j = r, i_{j+1} = r+1  \}  \cup \{ j \mid i_j = r+1, i_{j+1} = r   \}| \le k,$
\item \emph{$k$-visit} if,
for each $r \in \mathbb{Z}$, $| \{j \mid  i_j = r\}| \le k$. 
\end{itemize}
 Further, we say $M$ is \emph{$k$-turn} (respectively $k$-visit, $k$-crossing) if every accepting computation is $k$-turn (respectively $k$-visit, $k$-crossing). Also, $M$ is \emph{finite-turn} (respectively finite-visit, finite-crossing) if it is $k$-turn (respectively $k$-visit, $k$-crossing) for some $k \ge 0$.
\end{definition}

Given a class of machines $\MM$, let $\LL(\MM)$ be the family of languages accepted by machines in $\MM$, and let
$\SS(\MM)$ be the family of store languages of machines in $\MM$.

We denote the class of one-way nondeterministic finite automata by $\NFA$, and class of two-way nondeterministic finite automata by 
$2\NFA$, and as a short-form we denote the class of regular languages by $\REG$ and the class of recursively enumerable languages by $\RE$.
We denote the class of all finite-turn (respectively finite-visit, finite-crossing) Turing machines by $\ftNTM$ (respectively $\fvNTM$, $\fcNTM$).
We denote the class of all finite-turn (respectively finite-visit, finite-crossing) Turing machines with some number of reversal-bounded counters by $\ftNTCM$ (respectively $\fvNTCM$, $\fcNTCM$).

As noted in Section \ref{sec:intro}, every finite-turn Turing machines is finite-visit, but not vice versa (see Example \ref{TMwSwH*} below), and all finite-visit machines are finite-crossing, but not vice versa (each transition that stays on the same worktape cell increases the visit count but not the crossing count). In terms of languages accepted, however, it was shown in \cite{visitautomata} that $\LL(\ftNTM) \subsetneq \LL(\fvNTM) = \LL(\fcNTM)$.
In particular, $\LL(\ftNTM)$ is not closed under Kleene-*, but $\LL(\fcNTM)$ is closed under it. Also, as mentioned in Section \ref{sec:intro}, $\LL(\fcNTM)$ is equal to the family of languages generated by many types of finite-index grammar systems, such as finite-index matrix grammars.
Furthermore, $\LL(\fcNTCM)$ was studied in \cite{Harju} where it was shown that all languages in this family are semilinear (with an effective procedure) implying a decidable emptiness, membership, and finiteness problem.

\section{Store Languages of Finite-Visit and Finite-Crossing Turing Machines}
\label{storeFiniteVisit}

We give two examples of Turing machines and their store languages.

\begin{example}\label{TMwSw}
    Consider the $\NTM$ $M=(Q,\Sigma,\Gamma,\delta,q_0,F)$, where $\Sigma = \{a,b,\$\}$, $\overline{\Sigma} = \{a,b\}$, $\Gamma=\{a,b,\blank\}$, $F=\{q_4\}$, and $\delta$ contains the following transitions, for all $c\in\overline{\Sigma}$:
    \begin{align*}
            (q_0,c,0,\blank, q_1,c, {\rm R}), 
            (q_1,c,0,\blank,q_1,c, {\rm R}), 
            (q_1,\$,0,\blank,q_2,\blank, {\rm L}), 
            (q_2,\lambda,0,c,q_2,c, {\rm L}),\\
            (q_2,\lambda,0, \blank,q_3,\blank, {\rm R}),
            (q_3,c,0,c,q_3,c , {\rm R}), 
            (q_3,\emr,0,\blank,q_4,\blank, {\rm S}).
    \end{align*}
    Note that the third components are all $0$'s as it does not have any counters.
    This machine accepts the language $\{w\$w\mid w\in \overline{\Sigma}^+\}$, which is a well-known non-regular and non-context-free language \cite{HU}.
    It operates by copying the first part of the input (up to $\$$) to the tape, moving the read/write head back to the leftmost end of the tape, then matching the second part of the input with the contents of the tape.
    Hence this machine is 2-turn. It is also $3$-crossing, and $3$-visit.

    The store language of this machine, $S(M)$, is
$$\{q_0\blank\head\}
        \cup\{qx\blank\head\mid q\in\{q_1,q_3,q_4\},x\in\overline{\Sigma}^+\}
        \cup\{q_2\blank\head x\mid x\in\overline{\Sigma}^+\}
        \cup\{qx\head y\mid q\in\{q_2,q_3\},x\in\overline{\Sigma}^+,y\in\overline{\Sigma}^*\},$$
    which is a regular language.
\end{example}

This example is indicative of a general important property of store languages of $\ftNTM$ Turing machines.
We have the following from \cite{StoreLanguages}.
\begin{proposition}[\cite{StoreLanguages}]
 $\SS(\ftNTM) \subseteq \REG$.
\end{proposition}

Whether this was true for finite-visit and finite-crossing Turing machines was left open.
\begin{example}\label{TMwSwH*}
    Consider the following Turing machine
which is a  modified version of the machine given in Example \ref{TMwSw},
$M=(Q,\Sigma,\Gamma,\delta,q_0,F)$, where $\Sigma=\{a,b,\$,\#\}$, $\overline{\Sigma} = \{a,b\},\Gamma=\{a,b, \#,\blank\}$, $F=\{q_0\}$, and $\delta$ contains the following transitions, for all $c\in\overline{\Sigma}$ (leaving off all third components as they are all $0$'s):
    \begin{align*}
            (q_0,\lambda,\blank, q_1,\#, {\rm R}),
            (q_1,c,\blank, q_1,c, {\rm R}), 
            (q_1,\$,\blank, q_2,\blank, {\rm L}),
            (q_2,\lambda,c, q_2,c, {\rm  L}), \\
            (q_2,\lambda,\#, q_3,\#, {\rm R}), 
            (q_3,c,c, q_3,c, {\rm R}),
            (q_3,\#,\blank,q_0,\blank,{\rm S}).
    \end{align*}
    This machine accepts the language 
    $\{w\$w\# \mid w\in \overline{\Sigma}^*\}^*$.
    For each segment $u\$v$ of the input separated by $\#$'s, it writes $\#$, copies the string $u$ to the tape, moves the read/write head left until it reaches $\#$, then matches the string $v$ from the input with the contents of the tape to the right of the head.
    It then repeats this for the next segment of the input using a new section of the worktape.
        
    This machine is not finite-turn, since it makes $2n$ turns for an input with $n$ segments. However, $M$ is $4$-visit and $3$-crossing.
The store language of this machine is:
    \begin{multline*}
        S(M) = \{q_0x\blank\head\mid x\in(\#\overline{\Sigma}^*)^*\}
        \cup \{qx\blank\head\mid q\in\{q_1,q_3\},x\in(\#\overline{\Sigma}^*)^+\}\\
        \cup \{q_2x\head y\mid x\in(\#\overline{\Sigma}^*)^+, y\in\overline{\Sigma}^*\}
        \cup \{q_3x\head y\mid x\in(\#\overline{\Sigma}^*)^*\#\overline{\Sigma}^+, y\in\overline{\Sigma}^*\},
    \end{multline*}
    which is a regular language.
\end{example}
Next we shall prove that this is true in general for any finite-crossing Turing machine.


\begin{theorem} \label{fcNTM}
 $\SS(\ftNTM) \subseteq \REG$ and $\SS(\fvNTM) \subseteq \REG$.
\end{theorem}
\begin{proof}
It is enough to prove this for $\fcNTM$ since all finite-visit machines are finite-crossing.

    We use a construction that is somewhat similar to the one found in the proof of the previous proposition for $\ftNTM$ given in \cite{StoreLanguages}, however, it is quite a bit more complex.
    Let $M=(Q,\Sigma,\Gamma,\delta,q_0,F)$ be an $r$-crossing Turing machine. Let $k = 2r$. Notice that if $M$ is $r$-crossing, then each cell is involved in at most $2r=k$ crossings, where at most $r$ come from the cell to the left, and at most $r$ from the cell to the right.
    Let $\overline{\Gamma}=\{\overline{a}\mid a\in\Gamma\}$ be a new alphabet.
    Symbols in $\Gamma$ are said to be \emph{unmarked}, and those in $\overline{\Gamma}$ are said to be \emph{marked}.
We are going to construct a $2\NFA$ $M'$ that accepts words that represent a ``history'' of a computation of $M$. First, we will describe the alphabet $C$.   Intuitively, each symbol of $w$ to $M'$ is a column representing the history of a single cell of the worktape of $M$, and each component of the symbols of $w$ is a ``track'' (we call these tracks $0$ through $k$).
Track 0 will contain a representation of the worktape on some single point of time from the history of this computation, with the state in the first cell, and with a single marked symbol which represents the position of the read/write head at that point of time.
    In tracks 1 through $k$, the numbers and arrows encode a doubly linked list structure representing actions that $M$ performs on its tape given some input.

    Formally, define a new alphabet $C$ whose elements are of the form $(c_0,\dots,c_{k})$, where $c_0\in Q\cup\Gamma\cup \overline{\Gamma}$ and for each $j=1,\dots,k$, either $c_j=\varnothing$ or
    \[c_j\in(\Gamma \cup \overline{\Gamma}) \times\{\la,\ra,\varnothing\} \times\{1,\dots,k,\varnothing\} \times\{\la,\ra,\varnothing\} \times\{1,\dots,k,\varnothing\}.\]
    For such a $c_j=(x,s_d, s_t, p_d, p_t)$, $1 \le j \le k$, we call $x$ the tape symbol, $s_d$ the successor direction, $s_t$ the successor track, $p_d$ the predecessor direction, and $p_t$ the predecessor track. The new symbol $\varnothing$ is used as an ``unused'' marker.  
    
    Next, construct a $2\NFA$ $M' = (Q', C, \delta', q_0', F')$, where we will describe the construction of $Q'$ and $\delta'$ below, with input delimited by left and right end-markers $\triangleright$ and $\triangleleft$, \textit{i.e.}, the input to $M'$ is $\triangleright w\triangleleft$ where $w\in C^*$.
    \begin{figure}[h!]
        \centering
        \includegraphics[height=2.2in]{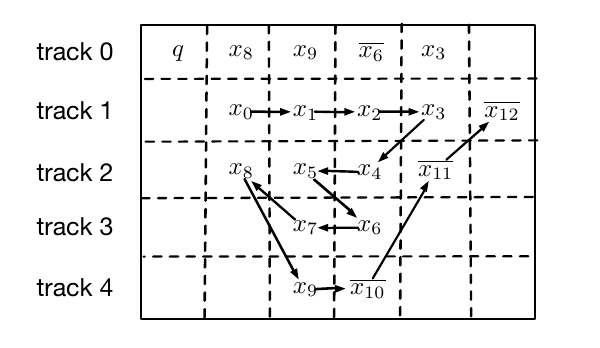}
        \caption{Example input to the $2\NFA$.}
        \label{2nfainput}
    \end{figure}
    Figure \ref{2nfainput} shows a visualization of an example input given to $M'$ which encodes a simulation of the following sequence of instructions of $M$. The encoding below shows the symbol to write in the current cell with a transition that moves to another cell (possibly after a sequence of stay transitions) along with the direction to move as superscript, corresponding to the computation from this example:
    \[x_0^\ra,\ x_1^\ra,\ x_2^\ra,\ x_3^\la,\ x_4^\la,\ x_5^\ra,\ x_6^\la,\ x_7^\la,\ x_8^\ra,\ x_9^\ra,\ x_{10}^\ra,\ x_{11}^\ra,\ x_{12},\]
    which corresponds to the following sequence of worktape contents of $M$:
\begin{align*}
        \blank\head,\ x_0\blank\head,\ x_0 x_1 \blank\head,\ x_0 x_1 x_2 \blank\head,\ x_0 x_1 x_2\head x_3,\ x_0 x_1 \head x_4 x_3 ,\ x_0 x_5 x_4 \head x_3, 
        x_0 x_5 \head x_6 x_3,\\ x_0\head x_7 x_6 x_3,\ x_8 x_7 \head x_6 x_3,\ \mathbf{ x_8 x_9 x_6 \head x_3},\ x_8 x_9 x_{10} x_3 \head,\ x_8 x_9 x_{10} x_{11} \blank\head,\ x_8 x_9 x_{10} x_{11} x_{12}\head.
\end{align*}
The bolded configuration corresponds to the figure.
    In Figure \ref{2nfainput}, we only visually show the arrows encoded by the second and third components of each track (the second component encodes arrow direction and the third component encodes the track of the arrow destination), and for each track there is an arrow encoded by the fourth and fifth components that exactly follows the shown arrow backwards.
    For example, the first three characters of the input corresponding to this visualization would be:
    \begin{gather*}
        (q,\varnothing,\varnothing,\varnothing,\varnothing) \\
        (x_8,(x_0,\ra,1,\varnothing,\varnothing),(x_8,\ra,4,\ra,3),\varnothing,\varnothing) \\
        (x_9,(x_1,\ra,1,\la,1),(x_5,\ra,3,\ra,2),(x_7,\la,2,\ra,3),(x_9,\ra,4,\la,2)).
    \end{gather*}
    The last component of the third character is encoded as $(x_9,\ra,4,\la,2)$, since the next node in the list is the fifth track of the column to the right, and the previous node in the list is the third track of the column to the left.

    Intuitively, as the arrows describe the movement between worktape cells of the computation, they fill in the next available unused track, from track $1$ to $k$. In this way, a linked list representing any $r$-crossing computation of $M$ can be stored in a string $w$ within these $k$ tracks. By following the arrows of the linked list, it is possible to follow the simulation of $M$ represented by $w$.
    At some nondeterministically chosen point in the simulation, track $0$ will be verified to contain the current tape contents at that point. Since that point is chosen nondeterministically, we get an encoding of each possible store language word on the first track with the read/write head marked. 
From then on, it continues the simulation in tracks $1$ to $k$, but we switch from using unmarked symbols to marked symbols.
    We can apply a generalized sequential machine (gsm) to extract the word contained in track $0$ (erasing blanks at the ends, and replacing the only marked symbol in track 1 $\overline{a}\in \overline{\Gamma}$ with $a\head$) to obtain a word in $S(M)$.

    The input to $M'$ is accepted if and only if $M'$ is able to successfully verify all of the following conditions:
    \begin{enumerate}
    
           \item Track 0 of the input word contains a word in $Q\Gamma^* \overline{\Gamma}\Gamma^*$, \textit{i.e.}, a word that could possibly be in $S(M)$ after applying the gsm mentioned above. In particular, track 0 contains exactly one marked symbol.
 
        \item In each column, tracks $1$ to $k$ are filled contiguously from track 1 downwards from the start to the end of the linked list.
        This implies that if some track $i$ is $\varnothing$, then tracks $i+1,\dots,k$ are also $\varnothing$. Thus, consider a node
        $(x,s_d,s_t,p_d,p_t)$ in the $i$th track, $i \ge 1$. Examining the column in direction $s_d$, the first track after $s_t$ with direction opposite $s_d$ (if it exists) must be to track $i+1$.
        
        
        \item There is exactly one node whose predecessor track and direction are both $\varnothing$, and this node is located in track 1. This node is called the beginning node. There is exactly one node whose successor track and direction are both $\varnothing$, and this node is located in the bottom-most non-empty track of its column. This node is called the end node. The predecessor and successor tracks and directions of all other nodes are not $\varnothing$.
        
        \item Consider a node  $(x,s_d, s_t, p_d, p_t)$ in the $i$th track. Then the predecessor track of the node in the $s_t$th track of the column to the left (if $s_d = \la$) or to the right (if $s_d = \ra$) must be set to $i$, and the predecessor direction of this node must be set to the opposite direction of $s_d$. The successor link from the predecessor of every node is verified similarly.

        
        

        \item In each column, if the symbol in track 0 is unmarked, then it must match the bottom-most unmarked symbol of the nodes in tracks $1,\dots,k$.
        
        
        \item The linked list represents an accepting computation of $M$ as follows: $M'$ starts by reading and storing the state $q$ written at the beginning of track 0 in the finite control of $M'$, and then $M'$ finds the beginning node in track 1. $M'$ then
 simulates transitions of $M$ by matching instructions with the symbols in the linked list rather than writing to a tape.
       $M'$ keeps track of the current track as it simulates.
                The most recent symbol that $M$ writes while staying on the same tape cell can also be tracked using the state of $M'$, and the symbols in the linked list represent the last symbol written in a cell before moving off of the tape cell (so the history does not contain any intermediate values written to a cell as it executes a sequence of stay transitions).
        The point in the linked list sequence where the first marked symbol occurs represents the point during the computation of $M$ where the word in track 0 matches the tape of $M$.
        This means that at some nondeterministically guessed point before the simulation of $M$ moves off of this tape cell (possibly in a sequence of stay transitions), the state of the simulated $M$ must match the state $q$ in track 0, and the symbol $M$ has just written to this tape cell (remembered in the state of $M'$) must match the marked symbol in track 0 of the same column.
        If this is successful, the simulation continues with only marked symbols until $M$ accepts, at which point $M'$ accepts.
    \end{enumerate}

    These conditions ensure that each input word represents a valid sequence of actions that $M$ could perform.
    In particular, they ensure that the directed graph described by the successor relation is a single directed path with no cycles.
    To see this, suppose that the list contains a cycle.
    Then there must be a node in the leftmost column of the cycle with both a predecessor and successor in the adjacent column to the right.
    Similarly, there is a node in the rightmost column of the cycle with both a predecessor and successor in the adjacent column to the left.
    Condition 2 implies that the successors of these nodes are in lower tracks than the predecessors of these nodes.
    Since these nodes are in the cycle, there must exist directed paths from the successor of the leftmost node to the predecessor of the rightmost node, and from the successor of the rightmost node to the predecessor of the leftmost node.
    However, these directed paths must cross at some point since they cannot extend beyond the leftmost and rightmost nodes, contradicting condition 2.
    
    Thus, $M'$ accepts words which describe an accepting computation of $M$ in tracks $1$ through $k$; and track $0$ contains the store word of $M$ in one specific configuration of this computation, after adjusting the head symbol, and possibly erasing blank symbols at the beginning and the end.
    It is known that the language accepted by any $2\NFA$ is regular \cite{HU}, hence $L(M')$ is a regular language.
    We can then apply the gsm that extracts the word from track $0$ of every word in $L(M')$, erasing blanks as appropriate and replacing $\overline{a}\in\overline{\Gamma}$ with $a\head$, to obtain $S(M)$.
    Since regular languages are closed under gsms \cite{HU}, it follows that $S(M)$ is a regular language.
\qed \end{proof}

Note that in all three cases of $\ftNTM$, $\fvNTM$, and $\fcNTM$, the store languages are ``essentially'' equal to the regular languages, except for the state at the beginning of each word, and the read/write head. For all three restrictions, given a regular language $R$, it is easy to construct a Turing machine $M$ that
would just put a word in $R$ on the worktape, then switch to some new fixed state $q$, and then switch to a final state. 
Let $S_q(M) = \{ u \mid q u \head \in S(M)\}$. For any subclass $\MM$ of Turing machines, let 
$\SS_{\rm state}(\MM)$ be the family of all $S_q (M)$ for $M \in \MM$ and states $q$ of $M$. 
Because $S(M)$ is always regular, $S_q(M)$ must be as well, for any state $q$ of $M$. This implies the following:
\begin{corollary}
$\SS_{\rm state}(\ftNTM) = \SS_{\rm state}(\fvNTM) = \SS_{\rm state}(\fcNTM) = \REG$.
\end{corollary}

Next, we consider each of $\ftNTCM$, $\fvNTCM$, and $\fcNTCM$. 
Each of these models uses $t$ counters, and the words of the store languages are of the form
$qx C_1^{i_1} \cdots C_t^{i_t}$ where $q$ is the current state, $x$ represents the contents of the worktape, and $i_j$ is the value of the $j$th counter, for $1 \le j \le t$ where $C_1, \ldots, C_t$ are fixed characters.

It was already shown that for finite-turn Turing machines augmented by reversal-bounded counters, all store languages can be accepted by $\NCM$ machines:
\begin{proposition} \cite{StoreLanguages}
$\SS(\ftNTCM) \subseteq \LL(\NCM)$.
\end{proposition}
But it remains open as to what occurs with finite-visit and finite-crossing Turing machines augmented by counters. We see that a similar construction to that of Theorem \ref{fcNTM} holds.
For this, we need two-way nondeterministic reversal-bounded multi-counter machines, denoted by $2\NCM$ \cite{Ibarra1978}.

\begin{proposition} \label{storeNCM}
$\SS(\fcNTCM) \subseteq \LL(\NCM)$ and $\SS(\fvNTCM) \subseteq \LL(\NCM)$.
\end{proposition}
\begin{proof}
We describe a modification of  the construction of Theorem \ref{fcNTM}. Let $M$ be an $\fcNTCM$ with $t$ counters.
Instead of building a $2\NFA$ $M'$, we build a $2\NCM$ with $3t$ counters. We refer to the counters  of $M'$ by the names
$c_1, \ldots, c_t, d_1, \ldots d_t, e_1, \ldots, e_t$. The machine operates just as in the simulation in the proof of Theorem \ref{fcNTM}, except the following:
\begin{itemize}

\item In point 1, track 0 of the input word contains a word in $Q \Gamma^* \overline{\Gamma} \Gamma^* C_1^* \cdots C_t^*$, \textit{i.e.}, a word that could possibly be in $S(M)$ after applying the gsm (that also fixes letters of $C_1, \ldots, C_t$ that appear in the store language) mentioned above. Namely, track 0 contains exactly one marked symbol.
\item  In point 6, it simulates the transitions of $M$ by using counters $c_1, \ldots, c_t$ of $M'$ to simulate the counters of $M$ faithfully until it hits the first marked symbol. Then in the nondeterministically guessed position where the marked symbol and state in track 0 matches the current worktape cell contents and state, it copies $c_j$ to both $d_j$ and $e_j$ for each $j$, $1 \le j \le t$ (so that now $d_j$ and $e_j$ have identical contents to what was in $c_j$, and $c_j$ is now $0$). It then continues the simulation using $d_j$, for each $j$ until it verifies that $M$ accepts. Finally, $M'$ verifies that the number of each of the characters $C_i$ at the end of the input word is equal to the value of the corresponding counter $e_i$.
\end{itemize}
 
 It follows that $S(M) = g( L(M'))$ where $g$ is the gsm.
 However, before applying $g$ to $L(M')$, we first apply another construction to $M'$, which is  a $2\NCM$. Notice that $M'$ satisfies the property that for every $x \in L(M')$, every accepting computation of $x$ crosses the boundary of each pair of adjacent cells of the {\bf input tape} a bounded number of times. This restriction of $2\NCM$ is a known model referred to by the name finite-crossing two-way nondeterministic reversal-bounded multi-counter machines \cite{Gurari}. Furthermore, it is known that each such machine can be converted to a one-way nondeterministic reversal-bounded multi-counter machine \cite{Gurari} that accepts the same language. Therefore, $M'$ can be converted to a (one-way) $\NCM$ $M''$ that accepts $L(M')$. Also, $\LL(\NCM)$ is closed under gsms \cite{Ibarra1978} and therefore $g(L(M'')) = S(M) \in \LL(\NCM)$.
\qed \end{proof}

As with the case without counters, given any $M \in \NCM$, we can build a machine $M' \in \ftNTCM$ (resp.\ $\fcNTCM, \fvNTCM$) that
copies the input to the worktape, switches to some new fixed state $q$, goes back to the beginning of the worktape, then simulates $M$ by reading the worktape to simulate the reading of the input, and using the counters faithfully. Then if we look at
$S_q (M') = \{ u \mid q u \head \in S(M')\}$, this is exactly $L(M)$. Again, considering any such class of machines $\MM$ and letting
$\SS_{\rm state}(\MM)$ be the family of all $S_q(M)$ where $M \in \MM$, we get the following:
\begin{corollary}
$\SS_{\rm state}(\ftNTCM) = \SS_{\rm state}(\fvNTCM) = \SS_{\rm state}(\fcNTCM) = \LL(\NCM)$.
\end{corollary}

Theorem \ref{fcNTM} and Proposition \ref{storeNCM}
immediately provide several applications using results in \cite{StoreLanguages,IbarraMcQuillanVerification}. First, denote the deterministic machines in $\fcNTM$ (respectively $\fvNTM$) by $\fcDTM$ (respectively $\fvDTM$). In this case,
a machine is deterministic if given any state $q$, worktape symbol $x$, and $a \in \Sigma \cup \{\emr\}$, there is at most one transition in the transition relation from $q$, with worktape symbol $x$, and on either input symbol $a$ or $\lambda$.
\begin{proposition}
The following are true:
\begin{itemize}

\item $\LL(\fcDTM)$ and $\LL(\fvDTM)$ are closed under right quotient by regular languages.
\item For any machine $M$ in any of $\fcNTM, \fvNTM$ and regular sets of store configurations $C \in \REG$, then both $\post_M^*(C)$ and $\pre_M^*(C)$ are regular languages, and a finite automaton can be constructed to accept each of them.
\item For any machine $M$ in any of $\fcNTCM, \fvNTCM$ and sets of store configurations $C \in \LL(\NCM)$, then both $\post_M^*(C)$ and $\pre_M^*(C)$ are in $\LL(\NCM)$, and an $\NCM$ can be constructed to accept each of them.
\item For any of $\fcNTM, \fvNTM, \fcNTCM, \fvNTCM$, the common configurations problem is decidable.
\end{itemize}
\label{implications}
\end{proposition}
\begin{proof}
For the first point, right quotient was studied in \cite{StoreLanguages}. There, a data store was called {\em readable} if it is possible for such a machine to have a ``stay'' that can be executed from any configuration (\textit{i.e.}\ any machine of that type could have an instruction designed to not to change its store; e.g. pushdown automata have a stay instruction that does not change its contents, and pushdown automata are allowed to define their transitions to use it from any state); and at any point, if the store contents are represented by $y$, it is possible to switch to a configuration where $y$ can be read one letter at a time, either from left-to-right (like a queue, which can repeatedly dequeue) or right-to-left (like a pushdown, which can repeatedly pop). Certainly both a finite-crossing worktape and a finite-visit worktape satisfy this property as they can read the contents until hitting blanks. Then given any readable store types, if we consider the class of all one-way nondeterministic machines $\MM$ with those readable store types where $\SS(\MM) \subseteq \REG$, and we let $\MM_D$ be the deterministic machines in $\MM$, then $\LL(\MM_D)$ are closed under right quotient with regular languages. Hence the first point follows.

For the second and third points, we need two technical properties studied in \cite{IbarraMcQuillanVerification} --- they are used to establish the connection between store languages and the $\pre^*$ and $\post^*$ operations.
We say that a machine model $\MM$ can be {\em  loaded} by a set of configurations $C$, if for all $M \in \MM$ with state set $Q$ and $C \subseteq \conf(M)$, there is a machine $M' \in \MM$ with $Q' \supseteq Q$ that, on input $q\gamma\$ x \$$ where $c=q\gamma \in \conf(M), q \in Q, x \in \Sigma^*$, and $\$$ is a new symbol, operates as follows:
\begin{enumerate}

\item  $M'$ reads $c = q\gamma$ while using states in $Q' \setminus Q$ and initializing the store configuration to $\gamma$.

\item After reading the first $\$$, it switches to state $q$ (the first state of $Q$ hit) and configuration $c$ if $c \in C$, and switches to a unique state $q_N \in Q' \setminus Q$ if $c \notin C$.

\item From $c$, $M'$ simulates $M$ on $x$ if and only if $c \in C$, as $M$ is defined on states of $Q$. 

\end{enumerate}
It is said that $\MM$ can be loaded by sets from some family $\LL$ if, for all $M \in \MM$ and $C \subseteq \conf(M)$ with $C \in \LL$, then $\MM$ can be loaded by $C$.

We say that a machine model $\MM$ can be {\em unloaded} by a set of configurations $C$ if, for all $M \in \MM$ with state set $Q$ and $C\subseteq \conf(M)$, there is a machine $M' \in \MM$ that, on input $c\$x\$, c = q\gamma \in \conf(M), q \in Q, x \in \Sigma^*$, operates as follows: 
\begin{enumerate}

\item $M'$ reads $q \gamma$ using states not in $Q$ while putting $\gamma$ on the stores.
\item Upon reading the first $\$$, it switches to state $q$ and configuration $c$.
\item It then simulates $M$ on $x$ starting from $c$.
\item Upon reading the second $\$$, it verifies that the current configuration of $M'$ is in $C$ and accepts only in this case. 
\end{enumerate} 
It is said that $\MM$ can be unloaded by sets from a family $\LL$ if, for all $M \in \MM$ and $C\subseteq \conf(M)$ with $C \in \LL$, then $C$ can be unloaded by $\MM$.

We will show that $\fcNTCM$ (respectively $\fvNTCM$) can be loaded and unloaded by sets from $\LL(\NCM)$ (this is already known for $\ftNTCM$, whose proof is essentially the same). The proofs for $\fcNTM,\fvNTM$, and $\ftNTM$ being loaded and unloaded by regular languages are simpler. First we will show it can be loaded by sets from $\LL(\NCM)$.
Consider an $\fvNTCM$ $M = (Q,\Sigma,\Gamma,q_0,F)$ with $t$ counters, and consider a language $C$ accepted by an $\NCM$ $M_C$ with $t_C$ counters. From this, we will describe how to build a $C$-loaded version $M' \in \fvNTCM$ of $M$ with $t+t_C$ counters. Throughout the simulation, the first $t$ counters will be used to simulate the counters of $M$, while the last $t_C$ counters will be used to simulate $M_C$.
The input will be verified to be of the form $q\gamma C_1^{i_1} \cdots C_t^{i_t} \$ x \$$ where $c = q\gamma C_1^{i_1} \cdots C_t^{i_t} \in \conf(M), q \in Q, \gamma \in \Gamma^*, x \in \Sigma^*$, which is processed as follows:
\begin{enumerate}

\item When reading $q\gamma$, $M'$ puts $\gamma$ on the worktape while marking the position of the read/write head.
\item When reading $C_1^{i_1} \cdots C_t^{i_t}$, $M'$ puts $i_j$ on counter $j$.

\item In parallel to executing steps 1 and 2, $M'$ verifies that $c \in C$ using the last $t_C$ counters.

\item  If $c \in C$, then it returns to the position of the encoded read/write head on the worktape and switches to state $q$ (if $c \notin C$, it rejects). From this configuration, $M'$ simulates $M$ on $x$ from $q$.
\end{enumerate}
Hence, $\fcNTCM$ can be loaded by sets from $\LL(\NCM)$.

For the unloaded case, again let $M = (Q,\Sigma,\Gamma,q_0,F)$ be an $\fvNTCM$ with $t$ counters, and consider an $\NCM$ $M_C$ with $t_C$ counters accepting $C$. From this, we can build a $C$-unloaded version $M' \in \fvNTCM$ of $M$ with $t+t_C$ counters. Again, throughout the simulation, the first $t$ counters will be used to simulate the counters of $M$, while the last $t_C$ counters will be used to simulate $M_C$.
The input will be verified to be of the form $q\gamma C_1^{i_1} \cdots C_t^{i_t} \$ x \$$ where $c = q\gamma C_1^{i_1} \cdots C_t^{i_t} \in \conf(M), q \in Q, \gamma \in \Gamma^*, x \in \Sigma^*$, which is processed as follows:
\begin{enumerate}

\item When reading $q\gamma $, it puts $\gamma$ on the worktape while marking the read/write head.
\item When reading $C_1^{i_1} \cdots C_t^{i_t}$, $M'$ puts $i_j$ on counter $j$. 
\item At the first $\$$, it then returns to the correct first position of the read/write head, and switches to $q$ and configuration $c$.
\item From there, it simulates $M$ starting at $q$ on $x$.
\item When it hits the second $\$$, it can verify that the current configuration is in $C$ by examining the current state $p$, scanning the current worktape $\gamma'$ from left-to-right, and the current counter contents $i_1', \ldots, i_t'$ to verify that 
$p \gamma' C_1^{i_1'} \cdots C_{t}^{i_t'} \in C$.
\end{enumerate}
Therefore, $\fcNTCM$ can be unloaded by sets from $\LL(\NCM)$.

From Theorem 20 of \cite{IbarraMcQuillanVerification}, it follows that the third property is true.

For the last property, Proposition 24 of \cite{IbarraMcQuillanVerification} says that for any machine model $\MM$ and language family $\LL$ such that:
\begin{itemize}

\item $\SS(\MM) \subseteq \LL$ with an effective construction,
\item $\LL$ has a decidable emptiness problem,
\item and $\LL$ is effectively closed under intersection,
\end{itemize}   then $\MM$ has a decidable common configurations problem. 
Both $\REG$ \cite{HU} and
$\LL(\NCM)$ \cite{Ibarra1978} have a decidable emptiness problem and are closed under intersection. The result follows
from Theorem \ref{fcNTM} and Proposition \ref{storeNCM}.
\qed \end{proof}

The results in Proposition \ref{implications} were previously unknown. The closure of the deterministic classes under right quotient by regular languages is interesting from a theoretical perspective. Furthermore, the result that $\pre_M^*(C)$  and $\post^*(C)$ can always be accepted by a finite automaton when $M$ is in either $\fcNTM$ or $\fvNTM$ and $C\in \REG$, is interesting towards the areas of model checking, reachability, and verification (similarly to the analogous result on pushdown automata \cite{PushdownVerification}), as is the decidability of the common configuration problem for all of these models even with reversal-bounded counters.

\section{Decidability Questions Regarding Turing Machine Restrictions}

Even though we now know that the store language of every $\ftNTM,\fvNTM,\fcNTM$ machine is regular, the question arises of whether we can determine if a given $\NTM$ is an $\ftNTM$ (respectively $\fvNTM,\fcNTM$) in the first place (which would guarantee that it has a regular store language). We prove next that this cannot be determined in general. 
The following result is shown for deterministic Turing machines (denoted by $\DTM$), and therefore it is also true for $\NTM$s.
\begin{proposition} It is undecidable whether a $\DTM$ is $0$-turn (resp., 
$0$-crossing, $1$-visit).
\label{prop1}
\end{proposition}
\begin{proof} We will show the $0$-turn case.
Let $M$ be a single-tape two-way deterministic Turing machine (with a combined input and two-way worktape as commonly defined \cite{HU}) with an initially blank tape.
Construct a one-way $\DTM$ $M'$ with input alphabet $\{a\}$ and
one read/write worktape. Then $M'$ on input $x$ operates as follows:
\begin{itemize}

\item If $x$ is $\lambda$, $M'$ accepts
without using the worktape tape.

\item If $x= a^n, n \ge 1$, $M'$ first simulates $M$ on the worktape.  There are two cases.
Case 1: $M$ does not halt.  Then $M'$ does not halt also.
Case 2: $M$ halts. Then $M'$ makes $n$ dummy
turns and accepts $x$.  
\end{itemize}
Clearly, if $M$ does not halt on blank tape, $M'$ would only accept $x = \lambda$ 
with $0$-turns. If $M$ halts, $M'$ would accept not only $a^n$ with $n$ additional turns.

It follows that $M'$ is $0$-turn if and only if $M$ does not halt on blank tape.
The result follows from the undecidability of the halting problem for 
$\DTM$s on an initially blank tape. The proof is identical for $0$-crossings (except in case 2 above, $M'$ would make $n$ dummy crosses).
For $1$-visit, $M'$ makes one visit immediately with the first move, and in case 2, $M'$ makes $n$ dummy visits.
\qed
\end{proof}

\begin{corollary} It is undecidable whether a $\DTM$  is an $\ftNTM$ (respectively, 
$\fvNTM, \fcNTM$).
\label{cor1}
\end{corollary}
This follows from the previous proposition, since $M'$ is finite-turn if and only if it is $0$-turn.

The next two propositions study decidability of two related problems: 1) Given  an $\NTM$ $M$, determining whether there is a word $x$ whereby all accepting computations on $x$ are $k$-turn (respectively $k$-visit, $k$-crossing), and 2) determining whether there is a word $x$ whereby there exists an accepting computation on $x$ that is $k$-turn (respectively $k$-visit, $k$-crossing). We will see that the first problem is undecidable, while the second is decidable. This second property implies that both problems are decidable for deterministic (or unambiguous, meaning there is at most one accepting computation on each word) machines.
\begin{proposition} It is undecidable, given an $\NTM$ $M$ and $k \ge 0$, whether there is an 
an input $x \in L(M)$ such that all accepting computations on $x$ are $k$-turn 
(respectively, $k$-crossing, $k\ge 1$-visit). 
\label{prop2}
\end{proposition}
\begin{proof}
Again, it is sufficient to show the $k$-turn case.

Let $M$ be a single-tape deterministic Turing machine (with a combined input and two-way worktape \cite{HU}) with an initially blank tape. Construct a one-way $\NTM$ $M'$ with input alphabet $\{a\}$.
First, $M'$ when given input $x$
makes $k$ turns on the worktape. Then it nondeterministically does one of the following:
\begin{enumerate}

\item $M'$ accepts $x$ if $x = a$ (taking $0$ turns).
\item $M'$ simulates $M$ on the worktape and if $M$ halts, $M'$ makes
a nondeterministically guessed number of turns and accepts $x$ if $x = a$ and rejects
all other strings.
\end{enumerate}
Clearly, $M'$ only accepts $a$, and not all accepting computations of $M$
on $x$ are within $k$-turns if and only if $M$ halts on blank tape.
\qed
\end{proof}

\begin{corollary}
\label{cor2} It is undecidable, given an $\NTM$ and $k \ge 0$, whether there is an 
an input $x$ such that only a finite number of accepting computations on $x$ are not $k$-turn (respectively not $k$-visit, not $k$-crossing).
\end{corollary}
This follows from the proof of the proposition above, since if $M$ halts, $M'$ has an unbounded
number of accepting computations with arbitrarily many turns.


Now we examine the second problem.
\begin{proposition} 
\label{weakly}
It is decidable, given an $\NTM$ $M$
and $k\ge 0$, whether there is an an input $x$ such that there is an accepting computation
on $x$ that is $k$-turn (respectively $k$-crossing, $k\ge 1$-visit). Furthermore, an $\ftNTM$ (respectively $\fvNTM, \fcNTM$) can be 
effectively constructed to accept all words $w$ where there is an accepting computation of $M$ on $w$ that is $k$-turn (respectively $k$-visit, $k$-crossing), and whereby every accepting computation is $k$-turn (respectively $k$-visit, $k$-crossing).

Both statements are also true for $\NTCM$ by constructing an $\ftNTCM$ (respectively $\fvNTCM, \fcNTCM$).
\end{proposition}
\begin{proof}
We prove the $k$-crossing case as it is the most complicated.
Notice that if an $\fcNTM$ (respectively $\fcNTCM$) can be constructed to accept all words $w$ where there is  an accepting computation of $w$ by $M$ that is $k$-crossing, then this implies that the problem is decidable, because emptiness is decidable for $\fcNTCM$ \cite{Harju}, and this machine is not empty if and only if $M$ accepts an input that is $k$-crossing.

Let $M$ be an $\NTCM$.  We construct another $\NTCM$ $M'$ that simulates $M$ as follows. Each cell of the worktape keeps either 
a worktape letter of $M$ or a non-negative integer, in an alternating fashion (between every two cells labelled by letters is a single non-negative integer).
Then $M'$ simulates $M$, but as it crosses to the right say, $M'$ moves right to an integer $m$ which it increases to $m+1$ before proceeding to the next cell. $M'$ also stores a variable called $S$ in its finite control which stores the maximum integer stored in any cell so far. All of these values only store integers up to $k$.

When $M$ accepts, $M'$ accepts if $S \le k$; else it rejects.
Therefore, if $M$ accepts $w$ with at most $k$ crossings, then so does $M'$, and every accepting computation of $M$ with at most $k$ crossings has a corresponding at most $k$ crossing accepting computation of $M'$. Any accepting computation using more than $k$ crossings does not have a corresponding accepting computation of $M'$. 

For the $k$-turn case, we only need to keep track of the number of turns so far in $S$. For the $k$-visit case, we need to store the number of visits to each cell as a second component of each cell along with the maximum visit number.
\qed
\end{proof}
To note, Proposition \ref{weakly} does not contradict Proposition \ref{prop2} because in Proposition \ref{prop2}, there could be additional accepting computations that are not $k$-turn, but the machine in Proposition \ref{weakly} has accepting computations corresponding with all accepting computations of $M$ with at most $k$ turns.

An $\NTCM$ $M$ is considered {\em unambiguous}, if for each $w \in \Sigma^*$, $M$ has at most one accepting computation on $w$.
The proof of the next result is similar to the one above, noting that
since the machine is unambiguous, for any input, there is at most 
one accepting computation. To note, every deterministic machine is unambiguous. This corollary is meant to contrast Proposition \ref{prop2}.
\begin{corollary}  It is decidable, given an unambiguous  $\NTCM$ $M$
and $k$, whether there is an an input $x$ such that the accepting computation
on $x$ is k-turn (resp., $k$-crossing, $k$-visit). 
\end{corollary}
It is also possible to construct an explicit example of a string $x$ with the desired property by testing non-emptiness, and then
if it is non-empty, testing non-emptiness of $L(M') \cap \{w\}$ ($M'$ in proof above) for all $w$ from the shortest until one accepts.

\section{Discussion Regarding Definitions of Turing Machine Restrictions}

Consider the following example of a non-finite-crossing $\NTM$.

\begin{example}\label{TMduplicate}
    Consider an $\NTM$ $M=(Q,\Sigma,\Gamma,\delta,q_0,F)$, where $\Sigma=\{a,b\}$, $\Gamma=\{a,b,a',b',\$,\blank\}$, $F=\{q_7\}$, and $\delta$ contains the following transitions, for all $x\in\{a,b\}$ and $y\in\{a,b,\$\}$:
    \begin{align*}
            (q_0,x,\blank,q_0,x, {\rm R}),
            (q_0,\emr,\blank,q_1,\$, {\rm L}),
            (q_1,\lambda,x,q_1,x, {\rm L}),
            (q_1,\lambda,\blank,q_2,\blank, {\rm R}),
            (q_2,\lambda,a,q_3,a', {\rm R}),
            (q_2,\lambda,b, q_4,b', {\rm R}),\\
            (q_3,\lambda,y,q_3,y, {\rm R}),
            (q_3,\lambda,\blank,q_5,a, {\rm L}),
            (q_4,\lambda,y,q_4,y, {\rm R}),
            (q_4,\lambda,\blank, q_5,b, {\rm L}),
            (q_5,\lambda,y, q_5,y, {\rm L}),
            (q_5,\lambda,x', q_6,x, {\rm R}),\\
            (q_6,\lambda,x,q_2,x, {\rm S}),
            (q_6,\lambda,\$, q_7,\$, {\rm S}).
    \end{align*}
    This machine accepts the simple regular language $\Sigma^+$, and when given an input $w$, produces $w\$w$ on the tape.
    It operates by first copying the input string to the tape followed by $\$$, then moving to the leftmost symbol on the tape.
    It marks the symbol using a prime symbol ($'$) and remembers it in the state, moves right to the first blank, writes the symbol, moves left to the marked symbol, erases the mark, moves right, then repeats until all symbols before $\$$ have been copied.
    
    Notice that when given an input of length $n$, performing the duplication on the tape requires the read/write head to cross the cell containing $\$$ at least $2n$ times. Thus this machine is not finite-crossing.
    
    Since this machine only reaches the final state $q_7$ after completing its operation, the only store configurations involving $q_7$ are those in which the tape contains a duplicated string, \textit{i.e.}, if we intersect $S(M)$ with the regular language $q_7\Gamma^*$, we obtain the language
    \[\{q_7w\$\head w\mid w\in\Sigma^+\}\]
    which is not regular. Thus $S(M)$ is not regular, since regular languages are closed under intersection \cite{HU}.
\end{example}
  We will refer to this example in discussing several points below.


In the definition of $k$-crossing machines we use in this paper, a machine is $k$-crossing if every accepting computation is $k$-crossing.
Elsewhere in the literature, it was defined that a machine $M$ is $k$-crossing if, for every $w \in L(M)$, there is some accepting computation that is $k$-crossing \cite{visitautomata}. Here, we call this notion {\em weakly $k$-crossing} (and {\em weakly finite-crossing}). The languages accepted by $k$-crossing and weakly $k$-crossing $\NTM$s coincide by Proposition \ref{weakly}. But  there is a large difference between the store languages of the two models. The following proof is reminiscent of Example \ref{TMduplicate}.
We use the definition of $S_q(M)$ and $\SS_{\rm state}(\MM)$ from Section \ref{storeFiniteVisit}. We add ``weak'' as a superscript to $\fcNTM$ to represent weakly finite-crossing machines.
\begin{proposition}
$\SS_{\rm state}(\fcNTM^{\rm weak}) = \RE$.  Furthermore, there are non-recursive languages that are store languages of weakly $0$-crossing $\NTM$s.
\end{proposition}
\begin{proof}
Let $M$ be a single-tape deterministic Turing machine (one combined input and read/write worktape) accepting $L$ over $\Sigma$ accepting an arbitrary $\RE$ language. Build $M'$ as follows: $M'$ does one of two things nondeterministically on input $w \in \Sigma^*$. Either
1) $M'$ simply reads $w$ and accepts, or
2) $M'$ copies $w$ to the worktape, switches to a new special state $q$ (which it only enters once), moves back to the first worktape cell before $w$, and then simulates $M$ on $w$, accepting if $M'$ accepts. Thus, $L(M')  = \Sigma^*$, and $M'$ is weakly $0$-crossing because for every $w\in L(M')$, there is some accepting computation (of type 1) that is $0$-crossing. But if one looks at $S_q(M')$, this is equal to $L(M)$.

For the second point, let $L$ be a non-recursive language accepted by a single-tape deterministic Turing machine $M$. By following the procedure of the first paragraph, we can construct $M$ such that $S_q(M') = L(M)$. Assume $S(M')$ were recursive. But then $(q^{-1}S(M'))(\head)^{-1} = S_q(M')$ would be as well since the recursive languages are closed under left and right quotient with a symbol, a contradiction.
\qed
\end{proof}

This illustrates that even if two classes of machines accept the same family of languages, this does not imply that their classes of store languages are equal. Indeed, in this case, finite-crossing $\NTM$ and weakly finite-crossing $\NTM$ accept the same family of languages. However, the first class of machines always have regular store languages, while the second has non-recursive store languages.
Notice that the machine $M'$ constructed in this proof uses nondeterminism. For unambiguous $\NTM$, where each word has at most one accepting computation, there is no difference in store languages between the two notions, and they are always regular in both cases. 

\section{Non-Bounded Visit and Crossing Turing Machines}

We have shown that for an $\NTM$, a constant bound on turns, crossings, or visits always produces a regular store language.
From Example \ref{TMduplicate}, we can see that given an input word of length $n$, a bound of $O(n)$ on the number of turns (respectively visits, crossings) can produce non-regular store languages.
This motivates us to consider other bounds on turns, crossings, and visits based on the length of the input to the Turing machine.

\begin{example}\label{TMlogn}
    Consider the Turing machine $M=(Q,\Sigma,\Gamma,\delta,q_0,F)$ where $\Sigma=\{a\}$, $\Gamma=\{a,b,\blank\}$, $F=\{q_7\}$, and $\delta$ contains the following transitions:
    \begin{align*}
            (q_0,a,\blank , q_0,a, {\rm R}), 
            (q_0,\emr,\blank , q_1,\blank, {\rm L}),
            (q_1,\lambda,a, q_1,a, {\rm L}), 
            (q_1,\lambda,\blank, q_2,\blank, {\rm R}),
            (q_2,\lambda,a,q_3,b, {\rm R}), 
            (q_2,\lambda,b ,q_2,b, {\rm R}),\\ 
            (q_3,\lambda,a,q_4,a, {\rm R}), 
            (q_3,\lambda,b, q_3,b, {\rm R}), 
            (q_3,\lambda,\blank,q_7,\blank, {\rm L}),
            (q_4,\lambda,a , q_5,b, {\rm R}),
            (q_4,\lambda,b,q_4,b, {\rm R}), 
            (q_4,\lambda,\blank,q_6,\blank, {\rm L}), \\
            (q_5,\lambda,a, q_4,a, {\rm R}),
            (q_5,\lambda,b,q_5,b, {\rm R}),
            (q_6,\lambda,a,q_6,a, {\rm L}), 
            (q_6,\lambda,b,q_6,b, {\rm L}), 
            (q_6,\lambda,\blank,q_2,\blank, {\rm R}).
    \end{align*}
    This machine accepts the language $\{a^{2^k}\mid k>0\}$. It operates by copying the input to the tape, moving the head back to the leftmost symbol, then sweeping from left to right while changing every other $a$ to $b$.
    After reaching the rightmost symbol, if the number of $a$'s seen was odd and greater than $1$, then the length of the input could not have been a power of two.
    If a single $a$ was seen, then the length of the input must have been a power of two, so $M$ accepts.
    If the number of $a$'s seen was even, then the head is returned to the leftmost symbol and another sweep is performed.
    
    Since each sweep halves the number of $a$'s on the tape, and $M$ only accepts when all $a$'s have been changed to $b$'s, this machine makes $2k+2$ turns for input of length $2^k$.
    Hence if the length of the input is $n$, then the number of turns is $O(\log n)$.
    Since this machine always sweeps the full length of the tape after a turn and has no transitions that stay on a tape cell, the number of crossings and visits are also both $O(\log n)$.
    
    This machine only reaches the final state $q_7$ after changing every $a$ on the tape to $b$, so if we intersect the store language $S(M)$ with the regular language $q_7\Gamma^*$, we obtain the language
    \[\{q_7b^{2^k}\head\mid k>0\}\]
    which is not regular. Thus $S(M)$ is not regular (nor semilinear \cite{GinsburgCFLs}).
\end{example}

From this example, it is clear that for a Turing machine given input of length $n$, a bound of $O(\log n)$ on turns, crossings, or visits is not sufficient to guarantee that the store language is regular.
This result can be strengthened as follows:
\begin{proposition}
Let $f : \mathbb{N} \rightarrow \mathbb{N}$ be an unbounded function, and let $M$ be an $\NTM$ which is not finite-turn, but on every input of length $n$, it makes at most $f(n)$ turns. Then there exists an $\NTM$ $M'$, which on every input of length $n$ makes at most $O(f(n))$ turns, but $S(M')$ is not regular.
\end{proposition}
\begin{proof}
Choose an arbitrary number $k_0$. Since $M$ is not finite-turn, it has some accepting computation on a string $w$, $|w| = n$, which makes at least $k_0$ turns. Denote the number of turns in this computation by $k$, where $k_0 \leq  k \leq f(n)$. Now construct $M'$ as the following: $M'$ has a two-track worktape, and it simulates a computation of $M$ on $w$ using the first track, but whenever $M$ would make a turn, $M'$ puts a new symbol $x$ on the second track. However, if this cell already contains the symbol $x$, $M'$ changes this symbol to a new symbol $y$, continues in the same direction as it was moving previously until it finds a cell with an empty second track, puts the symbol $x$ there, turns around, moves back to the cell with symbol $y$, changes this symbol back to $x$, and continues simulating the computation of $M$.

Note that after the simulation of $M$ ends, the second track of the tape of $M'$ contains exactly $k$ symbols $x$; and $M'$ has made $k$ turns so far.

$M'$ now moves its head to the left end of the tape (using up to one more turn), and sweeps the tape $k$ more times. In every sweep, the machine replaces one of the symbols $x$ on the second track by the symbol $y$, and appends one cell with the symbol $z$ in the second track at the very end of the tape. This continues until all symbols $x$ are erased, taking $2k$ more turns. Finally, $M'$ appends one more symbol $\#$ to the end of the tape and accepts. (This is a slightly modified variant of the construction from Example \ref{TMduplicate}.)

Observe that the computation of $M'$ has used $3k$ plus a constant number of turns, and since $k \le f(n)$, the number of turns is in $O(f(n))$. Further, the second track of the tape contains the string $x^k y^k \#$, potentially interleaved with blank spaces.

Let $h$ be a homomorphism which preserves the second track of the tape, and erases all tape cells which have an empty second track, and all non-tape symbols (state, head). Thus the image of the last configuration of $M'$ under $h$ is $x^k y^k \#$.

If $S(M')$ was regular, then so should be $L = (h(S(M'))$ intersected with $x^*y^*\#$. However, by the above construction, since we were able to choose an arbitrarily large $k$, we obtain an infinite number of strings $x^k y^k \#$ which should be in $L$. On the other hand, no string where the number of symbols $x$ and $y$ differ is in $L$, as the symbol $\#$ was only appended in the last step of $M'$. This is a contradiction with regularity of $L$, and thus also $S(M')$.

What this means is that for any function $f$, as long as some machine which makes $f(n)$ turns on an input of length $n$ exists, then $O(f(n))$ turns is not a sufficient bound to guarantee a regular store language. Example \ref{TMlogn} then demonstrates an $\NTM$ which makes $\log(n)$ turns on an input of length $n$.
\qed
\end{proof}
An open question is: does there exist an unbounded function $f$, such that every $f(n)$-turn-bounded machine is also constant-turn-bounded? If this was answered negatively, it would mean that the only possible limit on number of turns which guarantees a regular store language is constant.

\section{Store Languages of One-Way Multi-Counter Machines}
It is known that the store language of every $\NCM$ is in $\LL(\DCM)$, where $\DCM$ are the deterministic machines in $\NCM$ \cite{StoreLanguages}. But, how complex can the store languages of multi-counter machines get, when we do not start with the assumption that the counters are reversal-bounded? Certainly not all recursively enumerable languages can be store languages of multi-counter machines (even after taking the left quotient by a state) since the store language of every $t$ counter machine over state set $Q$ is a subset of $Q C_1^* \cdots C_t^*$, which is a letter-bounded language. It is known that even the simple language $\Sigma^*$ where $|\Sigma| \ge 2$ is not a bounded language \cite{GinsburgCFLs}, and therefore $\Sigma^*$  cannot be the store language of a multi-counter machine (even after taking any sort of quotient with the state set). However, we get non-recursive store languages. Let $\DCOUNTER(t)$ be the deterministic $\COUNTER(t)$ machines.
\begin{proposition} \label{nonREcounter}
$\SS(\DCOUNTER(2))$ contains non-recursive languages.
\end{proposition}
\begin{proof}
First, we will note the following well-known strategy. If we have a $2$-counter machine with counters called $A$ and $B$, with a number
$c$ in $A$, and $0$ in $B$, and we have a fixed number $d$, we can compute $dc$ and ultimately store it in $A$. We can do this by adding $d$ to $B$ for every decrease of $A$ by $1$, and then when $A$ is zero, move the contents of $B$ back to $A$.
Similarly, if we have a number $c$ in $A$ and $0$ in $B$, and we have a fixed number $d$, we can determine if $c$ is divisible by $d$, and compute $c$ divided by $d$ if it is. This can be done by adding $1$ to $B$ for every decrease of $d$ by $A$. If it is not divisible, then it is also clearly possible to recover $c$.

It is known that that there are non-recursive unary
languages that can be accepted by one-way deterministic 2-counter
machines \cite{HU}. Let $L$ be such a language accepted
by a machine $M$, where the counters are called $X$ and $Y$.  Construct a deterministic 2-counter machine $M'$ with
counters called $A$ and $B$. Then $M'$ on unary input $a^n \emr$
first computes on $A$ the number
$2^n$ using the doubling strategy above for every input letter read. 
When $M'$
reaches the end marker, $M'$ switches to special state $p$, and then
simulates $M$ in the following altered fashion. 
After simulating each move of $M$, counter $A$ of $M'$ will be of the form
$2^{c_1} 3^{c_2} 5^{c_3}$, where $c_1$ is the remaining number of $a$'s to read from the simulated computation,
$c_2$ is the current value in the simulated counter $X$, and
$c_3$ is the current value in the simulated counter $Y$.

Before simulating each transition of $M$, we can test whether there are any input letters left to read by checking if $A$ is divisible by $2$, and we can test whether any of the two simulated counters $X$ or $Y$ are non-zero by checking if $A$ is divisible by $3$ and by $5$. 
Then $M'$ can choose the appropriate transition, and can simulate the reading of an $a$ with a division by $2$, and an
increase to $X$ (resp.\ $Y$) by multiplying by $3$ (resp.\ $5$), and can simulate a decrease to $X$ (resp.\ $Y$) by dividing by $3$ (resp.\ $5$). Then $M'$ accepts if $M$ accepts.
Notice that $S(M') \cap p C_1^* C_2^* = pL'$, where $L' = \{a^{2^n} \mid a^n \in L\}$. Clearly,
$pL'$ is non-recursive, since recursive languages are closed under left
quotient with a symbol, and $L$ can easily be computed from $L'$.
\qed
\end{proof}

The above is true for deterministic machines. For nondeterministic machines, notice the following:
\begin{corollary}
There are non-recursive store languages of machines in $\COUNTER(2)$ accepting $\{\lambda\}$.
\end{corollary}
This is evident from the proof above as we can start by using nondeterminism to guess the value $n$.

\section{Store Languages of Two-Way Machines}

There has been comparatively little study of the store languages of two-way machines.
Two-way machines can be defined similarly to one-way machines, except each transition reads
a letter from $\Sigma$ or the left or right end marker
$\eml$ or $\emr$, and each transition has a final component that controls whether the input head position
moves to the left by one cell, stays in the same position, or moves to the right by one cell. Configurations
also include the input head position. 
We define acceptance by hitting a final state and falling off the input tape past the right end marker.
Please refer to the formal definitions for two-way machines \cite{StoreLanguages}.
We use the same model names as with one-way machines, but preface each class name with the number $2$ (like $2\NCM$).
As in the proof of Proposition \ref{storeNCM}, two-way input tape is called \emph{finite-crossing}
if in every accepting computation, the input head crosses the boundary of any two adjacent cells of the input
tape at most $c$ times for some $c$. In the next results, we also need the well-known class of one-way pushdown automata
(resp.\ augmented with reversal-bounded counters) $\NPDA$ (resp.\ $\NPCM$), and replace $N$ with $D$ for the deterministic restriction.

The only study of store languages of two-way machines that has been done in the past was in Section 4 of \cite{StoreLanguages}.
There, it was shown that all store languages of finite-crossing $2\NCM$s are in $\LL(\DCM)$ (that is, they can be accepted by one-way $\DCM$ machines). If the input is not finite crossing, it was shown that there are $2\DCM$ machines $M$ with one 1-reversal-bounded counters (called $2\DCM(1,1)$) accepting bounded languages such that the store languages are not in $\LL(\NPCM)$.
Similarly, there are $2\NCM(1,1)$ machines $M$ accepting unary languages such that $S(M)$ are not in $\LL(\NPCM)$ nor semilinear. Furthermore, some general connections between store languages of one-way and two-way acceptors was made. 
If a family of one-way acceptors $\MM_1$, and $\MM_2$ is the corresponding family of two-way acceptors, then 
$\MM_1$ has store languages contained in some family $\LL$ that is closed under homomorphism if and only if 
$\MM_2$ machines accepting finite languages have store languages contained in $\LL$.

In the next proposition, a $\DPDA^{1t}$ is a $\DPDA$ that is one-turn, which means that it does not push after popping, and finite-crossing refers to crossings on the two-way input tape.
\begin{proposition} \label{nonREtwoway}
There are non-recursive store languages of $2\DCOUNTER(1)$ machines accepting bounded languages.
Also, 
there are non-recursive store languages of machines in both finite-crossing $2\DCOUNTER(1)$ and finite-crossing $2\DPDA^{1t}$.
\end{proposition}
\begin{proof} 
To start, we show the result for $2\DCOUNTER(1)$.
Let $M$ be the 2-counter machine accepting a unary non-recursive
language $L$ over some alphabet $\{a\}$ in the proof of Proposition \ref{nonREcounter}. Let
$\#$ be a new symbol. We construct a two-way deterministic machine $M'$ with one counter
$C$ which, when given an input $a^n\#^d$  (for some $n, d \ge 1$), first converts $n$
to $2^n$ and stores it in counter $C$ using the doubling strategy (by using the
two-way input to act as another counter). This can be done if $d$ is large enough.
Then it simulates $M$ as in Proposition \ref{nonREcounter}, using $C$ and the input tape, rejecting if $d$ is not
large enough.
Notice that $M'$ accepts a bounded language that is a subset of $a^* \#^*$. 
For the case where $M'$ is finite-crossing, a similar construction is used which does not accept a bounded language.
Let $\$$ be a new symbol.
We construct a two-way deterministic machine $M'$ with one counter
$C$ which, when given an input $a^n\#^{d_1} \$ \cdots \#^{d_m} \$$  (for some $n, d_1, \ldots, d_m \ge 1$), first converts $n$
to $2^n$ and stores it in counter $C$ using the doubling strategy (by using the
two-way input to act as another counter). This can be done if $d$ is large enough.
$M'$ uses a new 
block $\#^{d_i} \$$ in simulating a counter every time it goes from increasing to decreasing. Clearly $M'$ is $c$-crossing for some $c$.

For $2\DPDA$, we use an arbitrary non-recursive language $L$ accepted by 
a single-tape $\DTM$ $Z$ with initial state $q_0$ that can only accept after an even number of configurations. On input $w$, $w \in L$ if and only if
there exists $x= ID_1 \# ID_3 \# \cdots \# ID_{2k-1} \#\# ID_{2k}^R \# ID_{2k-2}^R \# \cdots \# ID_2^R$ where $ID_1 \vdash \cdots \vdash ID_{2k}$ is an accepting computation of $Z$ on $w$ and $ID_1 = q_0w$. A $2\DPDA$ $M'$ with pushdown alphabet $\Gamma$ can be constructed which on input $x$,
pushes $w$ onto the pushdown, enters a special state $p$ (which it only enters once). Then it continues to push
$ID_1 \# ID_1 \# ID_3 \# ID_3 \# \cdots \# ID_{2k-1} \# ID_{2k-1}$ (doubling each configuration which it can do with the finite-crossing input tape) onto the pushdown until it reads $\# \#$. 
Then for each $i$ odd on the pushdown, as it is being popped, it is possible to verify that $ID_{i}$ can derive $ID_{i+1}$ and $ID_{i-1}$ can derive $ID_i$ by scanning each configuration after $ \# \# $ at most two times. Thus, $S(M') \cap p \Gamma^* = p L$, which is not recursive, otherwise $L$ would be as well.
\qed
\end{proof}
Note that for the $2\DPDA$, it is also possible to build $M'$ with two turns on the input and two turns on the pushdown by first verifying that odd configurations produce even configurations, and then verifying that even configurations produce odd configurations.

We showed in Proposition \ref{nonREtwoway}  that there exist non-recursive store languages of $2\DCOUNTER(1)$ machines.
 In the proof, the machines accept \emph{bounded languages} (subsets of $a^* \#^*$). A similar result is given in
 Proposition \ref{nonREtwoway}, where the machines are \emph{finite-crossing}, but in that case, the machines accept \emph{non-bounded languages}.
 However, when the machines are finite-crossing \emph{and} accept bounded languages, we can show that their store languages are
 recursive. In fact, we prove a much stronger result.
 We will need Theorem~2 in \cite{IbarraSeki}.
\begin{proposition} \cite{IbarraSeki}
A finite-crossing $2\DPCM$ accepting a bounded language is effectively equivalent to a $\DCM$. Hence, its emptiness problem is decidable.
\end{proposition}
Notice here that the $\DCM$ constructed is a one-way machine.

\begin{proposition} \label{finitecrossingDPCM}
The store language of every finite-crossing $2\DPCM$ accepting a bounded language is recursive.
\end{proposition}
\begin{proof}
Let $M$ be a finite-crossing $2\DPCM$ with $t$ reversal-bounded counters accepting a bounded language $L \subseteq \Sigma^*$, for some alphabet $\Sigma$. 
Every string in the store language $S(M)$ is of the form
\[
y = p z C_1^{n_1} \cdots C_t^{n_t},
\]
where $p$ is the current state, $z$ is the content of the stack, and $n_1, \ldots, n_t$ are the values of the counters.

We describe an algorithm to determine whether a given string $y$ is in $S(M)$. Construct a finite-crossing $2\DPCM$ $M_y$ with a pushdown stack $P$ and counters $C_1, \ldots, C_t$. $M_y$ stores the target configuration $y$ in its finite control.

On input $w$, $M_y$ simulates $M$ on $w$, maintaining in its finite control the simulated stack content as long as the height is less than or equal to $|z|$, and likewise keeping track of each counter $C_i$ as long as its value does not exceed $n_i$. It only begins using the actual pushdown stack $P$ when the stack height exceeds $|z|$, and uses counter $C_i$ only if its value exceeds $n_i$.

If, during this simulation, $M$ reaches the configuration $y$, then $M_y$ records this in its control. It then continues simulating $M$, and accepts if and only if $M$ accepts $w$. If $M$ never enters the configuration $y$, then $M_y$ rejects (regardless of whether $M$ accepts or not).

Clearly, the counters of $M_y$ are still reversal-bounded. Thus, $y \in S(M)$ if and only if the language $L(M_y)$ is non-empty. The result follows since the emptiness problem for $2\DPCM$ accepting  bounded languages is decidable.
\qed
\end{proof}

In Proposition \ref{nonREtwoway}, we showed that $\SS(2\DCOUNTER(1))$ contains non-recursive languages. However, if the counter is reversal-bounded, we have:
\begin{proposition} \label{twoway1rev}
$\SS(2\DCM(1))$ are recursive, and membership is decidable in $S(M)$ for each $M \in 2\DCM(1)$.
\end{proposition}
\begin{proof}
Let $M \in 2\DCM(1)$.
Every string in $S(M)$ (assuming $S(M)$ is non-empty) is of the
form $$y = pC_1^d,$$ where $p$ is a state of $M$ and $d$ represents the counter value.

We construct a two-way deterministic 1-counter machine $M_y$, whose counter is reversal-bounded, such that $y \in S(M)$ if and only if the language $L(M_y)$ is non-empty. The construction of $M_y$ is analogous to that in the proof of Proposition \ref{finitecrossingDPCM}.

 The result follows since the emptiness problem for $2\DCM(1)$ is
decidable \cite{IbarraJiang}.
\qed
\end{proof}

\begin{example}
Let $L = \{a^i b^j \mid i \ge 4, j \ge 2, i \text{ is divisible by } j\}$.
Clearly, $L$ can be accepted by a $2\DCM(1)$
$M$ which reads and stores $i$ in the counter and enters a special state $s$. 
Then it makes multiple passes on $b^j$ while decrementing $i$ and checks 
and accepts if $i$ is a multiple of $j$. Note that $M$ makes only one counter
reversal. Then $S(M)$ is not even a context-free language; otherwise, 
$S(M) \cap sC_1^+ = \{sC_1^i \mid i \ge 4, j \text{ is composite}\}$ would be context-free.
\end{example}

We do not know if Proposition \ref{twoway1rev} holds when $M$ is nondeterministic. However,
if $M$ accepts a bounded language, Proposition \ref{twoway1rev} holds, since the emptiness
problem for $2\NCM(1)$ 
that accepts bounded languages is decidable \cite{Ibarra2008}.
In general however, the following holds:
\begin{proposition} \label{twoway1revnon}
The 
membership problem for the store language for machines in $2\NCM(1)$ is decidable
if and only if the emptiness problem for $\NCM(1)$ is decidable.
\end{proposition}
\begin{proof}
Let $M$ be in $2\NCM(1)$ and $f$ be the unique accepting state of $M$, and assume without
loss of generality that all counters are zero before switching to $f$.  Then $f$ is in the
store language $S(M)$ of $M$ if and only if $L(M)$ is not empty. It follows
that if we can decide membership in $S(M)$, then we can decide the emptiness
of $L(M)$. Now suppose we can decide the emptiness of $L(M)$. Then, as illustrated
in the proof of Proposition \ref{twoway1rev}, we can decide membership in $S(M)$.
\qed
\end{proof}

Clearly, if the emptiness problem for machines in $2\NCM(1)$ is decidable (hence,
the membership problem for store languages is decidable), then the store
languages for machines in $\mathcal{M}$ are recursive. However, if the emptiness
problem is undecidable, it does not necessarily mean that there are store
languages that are not recursive.

The ideas in the proofs of Proposition \ref{twoway1revnon} can be applied to other 
counter machine models such as $2\DCM(2)$. 
It is known 
that the emptiness problem for these machines is undecidable, even when 
restricted to the special case when the machines accept letter-bounded
languages \cite{Ibarra1978}.
\begin{proposition}
Testing membership in $S(M)$ for $M \in 2\DCM(2)$ (even where $L(M)$ is letter-bounded) is undecidable.
\end{proposition}

\begin{remark}
In general, store languages can sometimes be more complicated than accepted languages, and can sometimes be simpler. 
The simpler case can be seen from the many cases where the store languages of models that accept more than regular languages only produce regular store languages. This is the case for one-way pushdown automata, and for $\fcNTM$.
The more complicated case occurs, for example if $M$ is a $2\NPCM$ or a $2\NPDA$, and therefore $L(M)$ is recursive. Indeed, to decide whether $y$ is in $L(M)$, 
construct  an $\NPCM$ $M_y$ which has $y$  in its state and simulates $M$ on 
$\lambda$ input,  Then $y$ is $L(M)$ if and only if $M_y$ accepts $\lambda$.  Since emptiness 
(hence, membership) for $\NPCM$ is decidable, $L(M)$ is recursive.  
However, store languages of $2\NPCM$, and even finite-crossing $2\DPDA$ where the pushdown is one-turn, and finite-crossing $2\COUNTER(1)$ are sometimes
non-recursive by Proposition \ref{nonREtwoway}.

To note that all languages accepted by $2\NCM$ are recursive, using the same argument above with decidability of emptiness for $\NCM$. Further, testing membership in a store language of $M \in 2\DCM(2)$ is undecidable by the proposition above.
Despite this, it is open whether or not all store languages of $2\DCM(2)$ (or even $2\NCM$) are recursive. It is possible that membership could be undecidable within a recursive store language.
\end{remark}

A two-way input tape is called finite-turn if the input head makes at most $c$
turns (reversals) on the input tape for some $c$. Clearly, finite-turn is
a special case of finite-crossing. We will show that Proposition \ref{nonREtwoway} does
not hold if finite-crossing is replaced by finite-turn.
We denote by $2\NCCM$ to be the same as $2\NPCM$ but the pushdown is replaced with an unrestricted counter.
\begin{proposition}
If $M$ is a finite-turn $2\NCCM$, then $S(M)$ is an effectively 
computable semilinear set that is in $\LL(\DCM)$.
\end{proposition}
\begin{proof}
It was shown in \cite{Ibarra1978} that if $M$ is a finite-turn 2$\NPCM$, then the Parikh image of $L(M)$ is an effectively 
computable semilinear set. Obviously, this result holds for finite-turn $2\NCCM$.
Let $\Sigma$ be the input alphabet of $M$ and $\#$ be a new symbol.
If $M$ has $t$ reversal-bounded counters,
then the strings in $S(M)$, if non-empty, are of the form
$1^s C_1^{i_1} C_2^{i_2} \cdots C_{t+1}^{i_{t+1}}$, where $s$ is the state, $i_1$ is the value
of the unrestricted counter, and $i_2, \dots, i_{t+1}$ are the values of the
reversal-bounded counters (we assume for convenience that the state set is $\{1,\ldots, m\}$ for some $m$).

Given $M$, we can construct a finite-turn $2\NCCM$ $M'$ which, when given input
$y = w \# 1^s C_1^{i_1} C_2^{i_2} \cdots C_{t+1}^{i_{t+1}}$, simulates $M$ on $w$. 
At some nondeterministically chosen point, $M'$ checks that the current configuration (i.e., the state and counter values) 
is represented on the
input. So it scans the input and it checks that the state is $s$, and the
counters have values $i_1, \ldots, i_{k+1}$. If not, it rejects. (Note that $M'$
can restore the values of the counters after checking). Then it continues the simulation
of $M$ on $w$. $M'$ accepts if $M$ accepts $w$. Then the Parikh image of 
$L(M')$ is an effectively computable semilinear set $S'$. It follows, by deleting the components
associated with the symbols in $\Sigma$ and $\#$ from the vectors generating $S'$,
we get the semilinear set for $S(M)$.
Finally, it is indeed known that all bounded-semilinear languages are in $\LL(\DCM)$ by \cite{IbarraSeki}.
\qed
\end{proof}

\section{Conclusions and Future Directions}

We have shown that store languages of finite-crossing and finite-visit Turing machines are regular, which was previously known for the less powerful finite-turn Turing machines. These two classes are important as they exactly accept the languages accepted by many types of finite-index grammars. From this, it is proven that the languages accepted by deterministic finite-visit and deterministic finite-crossing Turing machines are closed under right quotient with regular languages. Also it is determined that given a regular set of configurations $C$, the set of configurations that may follow in zero or more steps from $M$ is regular, and the set of configurations that may lead in zero or more steps to a configuration in $C$ is regular. 
Furthermore, adding reversal-bounded counters to either type of Turing machine produces store languages that can be accepted by machines with only reversal-bounded counters. This leads to more general reachability results for Turing machines (optionally with reversal-bounded counters). In addition, for either type of Turing machine (optionally with reversal-bounded counters), given two machines of this type, it is decidable whether they have some non-initial configuration in common --- a problem that has applications in fault tolerance.

Store languages of two-way machines are further studied, and many classes are shown to either have non-recursive store languages, or at least have store languages with an undecidable membership problem. In particular, store languages of two-way deterministic pushdown automata can be non-recursive.

In the future, it is worth considering if there is yet another larger class of Turing machines with regular store languages, or if there are other unrelated restrictions that can be imposed on Turing machines to produce regular store languages.
Examples \ref{TMduplicate} and \ref{TMlogn} present an interesting problem: how complex do the store languages become if we allow for a number of turns, crossings, or visits bounded by the size of the input.
There are also many other machine models whose store languages have yet to be studied.
A problem for future work is to characterize the store languages possible as a function of the number of turns, crossings, and visits, where it is larger than any constant.
Finally, in \cite{JALCTuringMachines} it was shown that it is decidable, given a finite-turn Turing machine $M$, whether the set of subwords of $L(M)$ is equal to $\Sigma^*$. The key property used there is that the store language of all finite-turn Turing machines are regular. Does a similar property hold for finite-crossing Turing machines, and hence many types of finite-index grammars?
Other properties of finite-crossing Turing machines are of interest. For example, it was recently shown that it is decidable to determine whether a given finite-turn Turing machine accepts a bounded language. Is this also true for finite-crossing Turing machines (and hence finite-index grammars)?

\bibliographystyle{elsarticle-num}
\bibliography{bounded}
\end{document}